\documentclass{article}
\usepackage{amsmath,amssymb}
\usepackage{amsthm}
\usepackage{subfigure}
\usepackage{tikz}
\usepackage[T1]{fontenc}
\usepackage[ruled,vlined,linesnumbered,english]{algorithm2e}

\newtheorem{lema}{Lemma}
\newtheorem{teorema}[lema]{Theorem}
\newtheorem{corolario}[lema]{Corollary}

\begin{document}

\title{The Predecessor-Existence Problem for $k$-Reversible Processes}

\author{Leonardo~I.~L.~Oliveira\\
Valmir~C.~Barbosa\\
\\
Programa de Engenharia de Sistemas e Computa\c c\~ao, COPPE\\
Universidade Federal do Rio de Janeiro\\
Caixa Postal 68511, 21941-972 Rio de Janeiro - RJ, Brazil\\
\\
F\'abio~Protti\thanks{Corresponding author (fabio@ic.uff.br).}\\
\\
Instituto de Computa\c c\~ao\\
Universidade Federal Fluminense\\
Rua Passo da P\'atria, 156, 24210-240 Niter\'oi - RJ, Brazil
}

\date{}

\maketitle

\begin{abstract}
For $k\geq 1$, we consider the graph dynamical system known as a $k$-reversible
process. In such process, each vertex in the graph has one of two possible
states at each discrete time. Each vertex changes its state between the present
time and the next if and only if it currently has at least $k$ neighbors in a
state different than its own. Given a $k$-reversible process and a configuration
of states assigned to the vertices, the \textsc{Predecessor Existence} problem
consists of determining whether this configuration can be generated by the
process from another configuration within exactly one time step. We can also
extend the problem by asking for the number of configurations from which a given
configuration is reachable within one time step. \textsc{Predecessor Existence}
can be solved in polynomial time for $k=1$, but for $k>1$ we show that it is
NP-complete. When the graph in question is a tree we show how to solve it in
$O(n)$ time and how to count the number of predecessor configurations in
$O(n^2)$ time. We also solve \textsc{Predecessor Existence} efficiently for the
specific case of $2$-reversible processes when the maximum degree of a vertex in
the graph is no greater than $3$. For this case we present an algorithm that
runs in $O(n)$ time.

\bigskip
\noindent
\textbf{Keywords:} $k$-reversible processes, Garden-of-Eden configurations,
Pre\-decessor-existence problem, Graph dynamical systems.
\end{abstract}

\newpage
\section{Introduction}\label{intr}

Let $G$ be a simple, undirected, finite graph with $n$ vertices and $m$ edges.
The set of vertices of $G$ is denoted by $V(G) = \{v_1, v_2, \ldots, v_n\}$ and
its set of edges is denoted by $E(G)$. A \textit{$k$-reversible} process on $G$
is an iterative process in which, at each discrete time $t$, each vertex in $G$ 
has one of two possible states. A state of a vertex is represented by an integer
belonging to the set $Q = \{-1,+1\}$ and each vertex has its state changed from
one time to the next if and only if it currently has at least $k$ neighbors in a
state different than its own, where $k$ is a positive integer.

Let $Y_t(v_i)$ be the state of vertex $v_i$ at time $t$. A
\textit{configuration} of states at time $t$ for the vertices in $V(G)$ is
denoted by $Y_t = (Y_t(v_1), Y_t(v_2), \ldots, Y_t(v_n))$. The one-step dynamics
in a $k$-reversible process for graph $G$ can be described by a function
$F_{G}^{k} : Q^n \to Q^n$ such that $Y_{t} = F_{G}^{k}(Y_{t-1})$ through a local
state update rule for each vertex $v_i$ given by
\begin{align}
Y_{t}(v_i) =
\left\{
\begin{array}{rl}
Y_{t-1}(v_i), &\mbox{if $v_i$ has fewer than $k$ neighbors in state $-Y_{t-1}(v_i)$}   \\
&\mbox{at time $t-1$;}   \\
-Y_{t-1}(v_i), &\mbox{otherwise.}\\
\end{array}
\right.
\end{align}

The motivation to study $k$-reversible processes is related to the analysis of
opinion dissemination in social networks. For example, suppose that a network is
modeled by a graph, each vertex representing a person and each edge between two
vertices indicating that the corresponding persons are friends. Suppose further
that state $-1$ represents disagreement on some issue and that the state $+1$
means agreement on the same issue. A $k$-reversible process is an approach to
model opinion dissemination when people are strongly influenced by the opinions
of their friends and the society they are part of. Notice that in this model we
are assuming that all people act in the same manner. A more complex approach
could assume, for example, distinct thresholds for each person or thresholds 
based on one's number of friends.

Note that $k$-reversible processes are examples of \textit{graph dynamical
systems}; more precisely, of \textit{synchronous dynamical systems}, which
extend the notion of a \textit{cellular automaton} to arbitrary graph
topologies. The study of graph dynamical systems and cellular automata is
multidisciplinary and related to several areas, like optics
\cite{phase-wrapping}, neural networks \cite{hopfield}, statistical mechanics
\cite{bootstrap-example}, as well as opinion \cite{opinion} and disease
dissemination \cite{disease-example}. Also in distributed computing there are
several studies regarding models of graph dynamical systems. An example is the
model of majority processes in which each vertex changes its state if and only
if at least half of its neighbors have a different state \cite{luccio}. An
application of this model is used for maintaining data consistency \cite{peleg}.

Most of the studies regarding $k$-reversible processes are related to the
\textsc{Minimum Conversion Set} problem in such processes. This problem
consists of determining the cardinality of the minimum set of vertices that, if
in state $+1$, lead all vertices in the graph also to state $+1$ after a finite
number of time steps. It has been proved that this problem is NP-hard for
$k > 1$ \cite{conversion_set}. There are also several interesting results
about this problem in the work by Dreyer \cite{processosReversiveis}, who also
presents some important results regarding the periodic behavior of
$k$-reversible processes, as well as upper bounds on the transient length that
precedes periodicity. Most of the results presented by Dreyer are based on
reductions from the so-called \textit{threshold processes}, which are broadly
studied by Goles and Olivos \cite{golesOlivos, golesOlivos2}. Another approach
to study the transient and periodic behavior of $k$-reversible processes is the
use of a specific energy function that leads to a much more intuitive proof of
the maximum period length and also better bounds on the transient length
\cite{dissertacao}.

A problem that arises in synchronous dynamical systems on graphs is the
so-called \textsc{Predecessor Existence} problem, defined as follows. Given one
such dynamical system and a configuration of states, the question is whether
this configuration can be generated from another configuration in a single time
step using the system's update rule. In the affirmative case, such configuration
is called a \textit{predecessor} of the one that was given initially. This
problem was studied by Sutner \cite{sutner} within the context of cellular
automata, where configurations lacking a predecessor configuration are known as
garden-of-Eden configurations, and was proved to be NP-complete for finite
cellular automata. NP-completeness results for related dynamical systems as well
as polynomial-time algorithms for some graph classes can also be found in the
literature \cite{pre}. An extension of the \textsc{Predecessor Existence}
problem is to count the number of predecessor configurations. This is also a
hard problem and has been proved to be $\#$P-complete \cite{preCount}. 

For $k$-reversible processes, we address these two problems in this paper. We
are interested in determining whether a configuration $Y_{t-1}$ exists for 
which $ Y_t = F_{G}^{k}(Y_{t-1})$. Because only time steps $t-1$ and $t$ matter,
for simplicity we denote $Y_{t-1}$ and $Y_{t}$ by $Y'$ and $Y$, respectively. We
henceforth denote this special case of \textsc{Predecessor Existence} by
\textsc{Pre($k$)}. We also consider the associated counting problem,
\textsc{\#Pre($k$)}, which asks for the number of predecessor configurations.
Our results include an NP-completeness proof for the general case of
$k$-reversible processes and polynomial-time algorithms for some particular
cases.

The remainder of the paper is organized as follows. In Section~\ref{pre1} we
show that \textsc{Pre($1$)} is polynomial-time solvable. In Section~\ref{np} we
provide an NP-completeness proof of \textsc{Pre($k$)} for $k > 1$. In
Section~\ref{trees} we describe two efficient algorithms for trees, one for
solving \textsc{Pre($k$)} and the other for solving \textsc{\#Pre($k$)}. In
Section~\ref{bounded} we show an efficient algorithm to solve \textsc{Pre($2$)}
for graphs with maximum degree no greater than $3$. Section~\ref{concl}
contains our conclusions.

\section{\boldmath Polynomial-time solvability of \textsc{Pre($1$)}}
\label{pre1}

For $k = 1$, if $Y'$ exists then any pair of neighbors $u$ and $v$ for which
$Y(u) = Y(v)$ also has $Y'(u) = Y'(v)$. Based on this observation, we start by
partitioning $G$ into connected subgraphs that are maximal with respect to the
property that each of them contains only vertices whose states in $Y$ are the
same. Clearly, all vertices in the same subgraph must have equal states also in
a predecessor of $Y$. Let us call each of such maximal connected subgraphs an
MCS.

Let $H$ be an MCS in which a vertex $v$ exists whose neighbors in $G$ all have
the same state as its own. In other words, all of $v$'s neighbors are also in
$H$. For this vertex, clearly there is no possibility other than $Y'(v) = Y(v)$.
Because there is only one choice of state for $v$ in a predecessor
configuration, we call both the vertex and its containing MCS \textit{locked}.
We refer to all other vertices and MCSs as being \textit{unlocked} (so there may
exist unlocked vertices in a locked MCS). We also say that any two MCSs are
neighbors whenever they contain vertices that are themselves neighbors.

\begin{teorema}
\label{k=1}
\textsc{Pre($1$)} is solved affirmatively if and only if the following two
conditions hold:
\begin{itemize}
\item No two locked MCSs are neighbors;
\item Every vertex in an unlocked MCS has at least one neighbor in another
unlocked MCS.
\end{itemize}
In this case, $Y'$ is obtained from $Y$ by changing the state of all vertices in
unlocked MCSs.
\end{teorema}

\begin{proof}
The case of a single (necessarily locked) MCS is trivial. If, on the other hand,
more than one MCS exists in $G$, then clearly the two conditions suffice for
$Y'$ to exist and be as stated: in one time step from $Y'$, the state of every
vertex in a locked MCS remains unchanged and that of every vertex in an unlocked
MCS changes, thus yielding $Y$.

It remains for necessity to be shown. We do this by noting that, should the
first condition fail and at least two locked MCSs be neighbors, any prospective
$Y'$ would have to differ from $Y$ in all vertices of each of these MCSs, but
the presence of locked vertices in them would make it impossible for $Y$ to be
obtained in one time step. Should the second condition be the one to fail and at
least one vertex in an unlocked MCS have neighbors outside its MCS only in
locked MCSs, any prospective $Y'$ would have to differ from $Y$ in all vertices
of such an unlocked MCS. Once again it would be impossible to obtain $Y$ in one
time step due to the locked vertices. It follows that both conditions are
necessary for $Y'$ to exist.
\end{proof}

By Theorem \ref{k=1}, we can easily solve \textsc{Pre($1$)} in $O(n+m)$ time.

\section{\boldmath NP-completeness of \textsc{Pre($k$)} for $k>1$}\label{np}

We first note that the NP-completeness proof of \textsc{Predecessor Existence}
for finite cellular automata \cite{sutner} cannot be directly extended to
$k$-reversible processes in graphs for $k>1$. Sutner's proof only shows that
there exist finite cellular automata for which \textsc{Predecessor Existence} is
NP-complete. In other words, it depends on the vertex update rule being used in
the cellular automaton. A different approach is then needed.

We present a reduction from a satisfiability problem known as \textsc{3Sat
Exact\-ly-Two}. This problem is the variation of the \textsc{3Sat} problem in
which each clause must be satisfied by exactly two positive literals. We start
by proving that \textsc{3Sat Exactly-Two} is NP-complete. 

\begin{lema}
\label{lema:3sat}
\textsc{3Sat Exactly-Two} is NP-complete.
\end{lema}

\begin{proof}
The problem is trivially in NP. We proceed with the reduction from another
variation of the \textsc{3Sat} problem, known as \textsc{3Sat Exactly-One}
\cite{jhonson}, which is NP-complete and asks whether there exists an assignment
of variables satisfying each clause by exactly one positive literal.

The reduction is simple and consists of inverting all literals in all clauses of
an instance $S$ of \textsc{3Sat Exactly-One}, resulting in an instance $S'$ of
\textsc{3Sat Exactly-Two}. It is easy to check that a solution for $S$ directly
gives a solution for $S'$, and conversely, a solution for $S'$ directly gives a
solution for $S$.
\end{proof}

\begin{teorema}
\label{np-completude}
\textsc{Pre($k$)} is NP-complete for $k>1$.
\end{teorema}

\begin{proof}
Given two configurations $Y$ and $Y'$, verifying whether $Y'$ is a predecessor
configuration of $Y$ is straightforward and can be done by simulating one step
of the $k$-reversible process starting with configuration $Y'$. Simulating one
step of the process takes $O(n + m)$ time and the final comparison between the
resulting configuration and $Y$ takes $O(n)$ time; thus, \textsc{Pre($k$)} is in
NP. The remainder of the proof is a reduction from \textsc{3Sat Exactly-Two},
which by Lemma \ref{lema:3sat} is NP-complete.

Let $S$ be an arbitrary instance of \textsc{3Sat Exactly-Two} with the $M$
clauses $c_{1}, c_{2}, \dots, c_{M}$ and the $N$ variables $x_{1}, x_{2},
\dots, x_{N}$. We construct an instance $(G, Y)$ of \textsc{Pre$(k)$} from $S$
as follows.

Vertex set $V(G)$ is the union of:
\begin{itemize}
\item $\{x_{i}, \neg{x_{i}}\}$, for each variable $x_{i}$ in $S$;
\item $\{z_{i}, z'_{i}\}$, for each variable $x_{i}$ in $S$;
\item $\{u_{i,1}, \dots , u_{i,2k-3}\}$, for each variable $x_{i}$ in $S$;
\item $\{p_{i,1}, \dots , p_{i,2k-3}\}$, for each variable $x_{i}$ in $S$;
\item $\{w_{i,1}, \dots , w_{i,k-2}\}$, for each variable $x_{i}$ in $S$,
provided $k > 2$;
\item $\{w'_{i,1}, \dots , w'_{i,k-2}\}$, for each variable $x_{i}$ in $S$,
provided $k > 2$;    
\item $\{c_{i}, c'_{i}\}$, for each clause $c_{i}$ in $S$;
\item $\{b_{i,1}, \dots , b_{i,k-2}\}$, for each clause $c_{i}$ in $S$, provided
$k > 2$;
\item $\{b'_{i,1}, \dots , b'_{i,k-1}\}$, for each clause $c_{i}$ in $S$. 
\end{itemize}

Vertices $x_{i}$ and $\neg{x_{i}}$ are called \textit{literal vertices} and
vertices $c_{i}$ and $c'_{i}$ are called \textit{clause vertices}. If $x$ is a
neighbor of $u$ and $x$ is a literal vertex, then we say that $x$ is a literal
neighbor of $u$. Similarly, if $x$ is a neighbor of $u$ and $x$ is a clause
vertex, then $x$ is a clause neighbor of $u$.

Edge set $E(G)$ is the union of:
\begin{itemize}
\item $\{(x_i,z_i), (x_i,z'_i), (\neg{x_i},z_i), (\neg{x_i},z'_i)\}$, for each
variable $x_{i}$ in $S$;
\item $\{(x_i,u_{i,1}), \dots, (x_i,u_{i,2k-3})\}$, for each variable $x_{i}$
in $S$;
\item $\{(\neg{x_i},p_{i,1}), \dots, (\neg{x_i},p_{i,2k-3})\}$, for each
variable $x_{i}$ in $S$;
\item $\{(z_i,w_{i,1}), \dots , (z_i,w_{i,k-2})\}$, for each variable $x_{i}$ in
$S$, provided $k > 2$;
\item $\{(z'_i,w'_{i,1}), \dots, (z'_i,w'_{i,k-2})\}$, for each variable $x_{i}$
in $S$, provided $k > 2$;
\item $\{(c'_{j}, b'_{j,1}), \dots, (c'_{j}, b'_{j,k-1})\}$, for each clause
$c_{j}$ in $S$;
\item $\{(c_{j}, b_{j,1}), \dots, (c_{j}, b_{j,k-2})\}$, for each clause
$c_{j}$ in $S$, provided $k > 2$;
\item $\{(c_{j}, x_{i}), (c'_{j}, x_{i})\}$, for each literal $x_i$ occurring in
clause $c_j$;
\item $\{(c_{j}, \neg{x_{i}}), (c'_{j}, \neg{x_{i}})\}$, for each literal
$\neg{x_i}$ occurring in clause $c_j$.
\end{itemize}

We finish the construction by defining the target configuration $Y$:
\begin{itemize} 
\item $Y(x_{i}) = Y(\neg{x_{i}}) =  +1$,  for $1 \leq i \leq N$; 
\item $Y(z_{i}) = +1$, $Y(z'_{i}) = -1$, for $1 \leq i \leq N$; 
\item $Y(u_{i,j}) = Y(p_{i,j}) = +1$, for $1 \leq i \leq N$ and $1 \leq j \leq
k-1$;
\item $Y(u_{i,j}) = Y(p_{i,j}) = -1$, for $1 \leq i \leq N$ and $k \leq j \leq
2k-3$, provided $k > 2$;
\item $Y(w_{i,j}) = -1$, $Y(w'_{i,j}) = +1$, for $1 \leq i \leq N$ and $1 \leq j\leq k-2$, provided $k > 2$; 
\item $Y(c_{i}) = +1$, $Y(c'_{i}) = -1$, for $1 \leq i \leq M$;
\item $Y(b_{i,j}) = -1$, for $1 \leq i \leq M$ and $1 \leq j \leq k-2$, provided
$k > 2$; 
\item $Y(b'_{i,j}) = -1$, for $1 \leq i \leq M$ and $1 \leq j \leq k-1$,
provided $k > 2$. 
\end{itemize}

Figure \ref{fig:config_y} illustrates the case of $k = 3$, $M = 1$, and $N = 3$,
the single clause being $c_{1} = x_{1} \vee \neg{x_{2}} \vee \neg{x_{3}}$.

\begin{figure}[t]
\begin{tikzpicture}
[scale=.7,auto=left,inner sep=0pt,minimum size=0.55cm]
\node[draw,fill=lightgray,circle] (n6)  at (1,10) {$x_{1}$};
\node[draw,fill=lightgray,circle] (n4)  at (3,10) {$\neg{x_{1}}$};
\node[draw,fill=lightgray,circle] (n5)  at (7,10) {$x_{2}$};
\node[draw,fill=lightgray,circle] (n1)  at (9,10) {$\neg{x_{2}}$};
\node[draw,fill=lightgray,circle] (n2)  at (13,10) {$x_{3}$};
\node[draw,fill=lightgray,circle] (n3)  at (15,10) {$\neg{x_{3}}$};
\node[draw,fill=lightgray,circle] (n7)  at (5,3) {${c_{1}}$};
\node[draw,fill=white,circle] (n8)  at (10,3) {${c'_{1}}$};
\node[draw,fill=white,circle] (n9)  at (9,1) {${b'_{1,1}}$};
\node[draw,fill=white,circle] (n10)  at (11,1) {${b'_{1,2}}$};
\node[draw,fill=white,circle] (n11)  at (5,1) {${b_{1,1}}$};
\node[draw,fill=lightgray,circle] (n12)  at (0,11) {$u_{1,1}$};
\node[draw,fill=lightgray,circle] (n13)  at (0,10) {$u_{1,2}$};
\node[draw,fill=white,circle] (n100)  at (0,9) {$u_{1,3}$};
\node[draw,fill=lightgray,circle] (n14)  at (4,11) {$p_{1,1}$};
\node[draw,fill=lightgray,circle] (n15)  at (4,10) {$p_{1,2}$};
\node[draw,fill=white,circle] (n101)  at (4,9) {$p_{1,3}$};
\node[draw,fill=lightgray,circle] (n16)  at (6,11) {$u_{2,1}$};
\node[draw,fill=lightgray,circle] (n17)  at (6,10) {$u_{2,2}$};
\node[draw,fill=white,circle] (n102)  at (6,9) {$u_{2,3}$};
\node[draw,fill=lightgray,circle] (n18)  at (10,11) {$p_{2,1}$};
\node[draw,fill=lightgray,circle] (n19)  at (10,10) {$p_{2,2}$};
\node[draw,fill=white,circle] (n103)  at (10,9) {$p_{2,3}$};
\node[draw,fill=lightgray,circle] (n20)  at (12,11) {$u_{3,1}$};
\node[draw,fill=lightgray,circle] (n21)  at (12,10) {$u_{3,2}$};
\node[draw,fill=white,circle] (n104)  at (12,9) {$u_{3,3}$};
\node[draw,fill=lightgray,circle] (n22)  at (16,11) {$p_{3,1}$};
\node[draw,fill=lightgray,circle] (n23)  at (16,10) {$p_{3,2}$};
\node[draw,fill=white,circle] (n105)  at (16,9) {$p_{3,3}$};
\node[draw,fill=lightgray,circle] (n24)  at (1,12) {$z_{1} $};
\node[draw,fill=white,circle] (n25)  at (3,12) {$z'_{1}$};
\node[draw,fill=lightgray,circle] (n26)  at (7,12) {$z_{2}$};
\node[draw,fill=white,circle] (n27)  at (9,12) {$z'_{2}$};
\node[draw,fill=lightgray,circle] (n28)  at (13,12) {$z_{3}$};
\node[draw,fill=white,circle] (n29)  at (15,12) {$z'_{3}$};
\node[draw,fill=white,circle] (n30)  at (1,14) {$w_{1,1}$}; 
\node[draw,fill=white,circle] (n31)  at (7,14) {$w_{2,1}$}; 
\node[draw,fill=white,circle] (n32)  at (13,14) {$w_{3,1}$}; 
\node[draw,fill=lightgray,circle] (n33)  at (3,14) {$w'_{1,1}$}; 
\node[draw,fill=lightgray,circle] (n34)  at (9,14) {$w'_{2,1}$}; 
\node[draw,fill=lightgray,circle] (n35)  at (15,14) {$w'_{3,1}$};  
\foreach \from/\to in {n100/n6, n101/n4, n102/n5, n103/n1, n104/n2, n105/n3,
n8/n9, n8/n10, n11/n7,n7/n6,n7/n3,n7/n1,n8/n6,n8/n3,n8/n1,n24/n6,n24/n4,n25/n6,
n25/n4, n26/n5, n26/n1, n27/n5, n27/n1, n28/n2,n28/n3, n29/n2, n29/n3, n12/n6,
n13/n6, n14/n4, n15/n4, n16/n5, n17/n5, n18/n1, n19/n1, n20/n2, n21/n2, n22/n3,
n23/n3, n30/n24, n31/n26, n32/n28, n33/n25, n34/n27, n35/n29}
\draw (\from) -- (\to);
\end{tikzpicture}
\caption{Graph $G$ in the instance of \textsc{Pre(3)} having $M = 1$ and $N = 3$
for which $c_{1} = x_{1} \vee \neg{x_{2}} \vee \neg{x_{3}}$. Shaded circles
indicate state $+1$ in configuration $Y$; empty circles indicate state $-1$.}
\label{fig:config_y}
\end{figure}
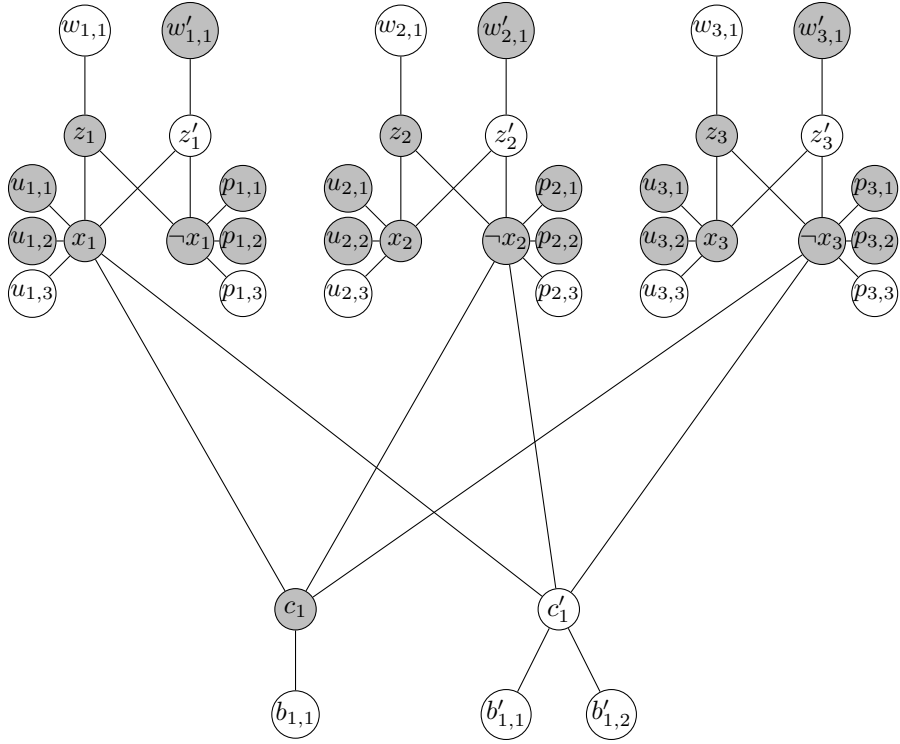

For each variable in $S$, at most $6k - 6$ vertices are created, and for each
clause at most $2k - 1$ vertices, resulting in at most $n = N(6k-6) + M(2k-1)$
vertices. Likewise, the total number of edges is at most $m = N(6k-6) +
M(2k+3)$. Since $k$ is a constant, we have a polynomial-time reduction.
 
We proceed by showing that if $S$ is satisfiable then $Y$ has at least one
predecessor configuration. In fact, given any satisfying truth assignment for
$S$, we can construct a predecessor configuration $Y'$ of $Y$ in the following
manner:
\begin{itemize}
\item $Y'(x_{i}) = +1$, if variable $x_{i}$ is true in the given assignment; 
\item $Y'(x_{i}) = -1$, if variable $x_{i}$ is false in the given assignment; 
\item $Y'(\neg{x_{i}}) = -Y'(x_{i})$;
\item $Y'({z_{i}}) = +1$, $Y'(z'_{i}) = -1$, for $1\leq i \leq N$;
\item $Y'(u_{i,j}) = Y'(p_{i,j}) = +1$, for $1 \leq i \leq N$ and $1 \leq j
\leq k-1$;
\item $Y'(u_{i,j}) = Y'(p_{i,j}) = -1$, for $1 \leq i \leq N$ and $k \leq j
\leq 2k-3$;
\item $Y'(w_{i,j}) = -1$,  $Y'(w'_{i,j}) = +1$, for $1 \leq i \leq N$ and $1
\leq j \leq k-2$, provided $k > 2$;
\item $Y'(c_{i}) = Y'(c'_{i}) = +1$, for $1 \leq i \leq M$;
\item $Y'(b_{i,j}) = -1$, for $1 \leq i \leq M$ and $1 \leq j \leq k-2$,
provided $k > 2$;
\item $Y'(b'_{i,j}) = -1$, for $1 \leq i \leq M$ and $1 \leq j \leq k-1$.
\end{itemize}

Figure \ref{fig:config_ylinha} shows predecessor configuration $Y'$ for the $Y$
given in Figure \ref{fig:config_y}. As an example, we have let all of $x_{1}$,
$\neg{x_{2}}$, and $x_{3}$ be true in the satisfying assignment.

By construction, and considering configuration $Y'$ as described above, each
vertex $x_{i}$ or $\neg{x_{i}}$ has at least $k$ neighbors in state $+1$ in
$Y'$ and exactly $k-1$ neighbors in state $-1$. This condition is sufficient to
guarantee that every literal vertex reaches state $+1$ in one time step,
regardless of its state in configuration $Y'$.

Each of vertices $u_{i,j}$, $p_{i,j}$, $w_{i,j}$, $w'_{i,j}$ , $b_{i,j}$, and
$b'_{i,j}$ has only one neighbor, and will obviously keep its state in the next
configuration. Each of vertices $z_{i}$ and $z'_{i}$ has at most $k - 1$
neighbors in the opposite state, since vertices $x_{i}$ and $\neg{x_{i}}$ have
mutually opposite states. So $z_{i}$ and $z'_{i}$ remain unchanged as well.

In order for configuration $Y$ to be obtained after one time step, in
configuration $Y'$ no vertex $c_{i}$ can have more than one literal neighbor in
state $-1$, which would lead to state $-1$ for $c_i$ in the next configuration.
Since $S$ is satisfiable, there are exactly two positive literals for each
clause in $S$, and by construction of $Y'$, each vertex $c_{i}$ has exactly one
literal neighbor in state $-1$. Similarly, every vertex $c'_{i}$ needs at least
one literal neighbor in state $-1$, and as the construction guarantees, $c'_{i}$
has exactly one literal neighbor in state $-1$.

\begin{figure}[t]
\begin{tikzpicture}
[scale=.7,auto=left,inner sep=0pt,minimum size=0.55cm]
\node[draw,fill=lightgray,circle] (n6)  at (1,10) {$x_{1}$};
\node[draw,fill=white,circle] (n4)  at (3,10) {$\neg{x_{1}}$};
\node[draw,fill=white,circle] (n5)  at (7,10) {$x_{2}$};
\node[draw,fill=lightgray,circle] (n1)  at (9,10) {$\neg{x_{2}}$};
\node[draw,fill=lightgray,circle] (n2)  at (13,10) {$x_{3}$};
\node[draw,fill=white,circle] (n3)  at (15,10) {$\neg{x_{3}}$};
\node[draw,fill=lightgray,circle] (n7)  at (5,3) {${c_{1}}$};
\node[draw,fill=lightgray,circle] (n8)  at (10,3) {${c'_{1}}$};
\node[draw,fill=white,circle] (n9)  at (9,1) {${b'_{1,1}}$};
\node[draw,fill=white,circle] (n10)  at (11,1) {${b'_{1,2}}$};
\node[draw,fill=white,circle] (n11)  at (5,1) {${b_{1,1}}$};
\node[draw,fill=lightgray,circle] (n12)  at (0,11) {$u_{1,1}$};
\node[draw,fill=lightgray,circle] (n13)  at (0,10) {$u_{1,2}$};
\node[draw,fill=white,circle] (n100)  at (0,9) {$u_{1,3}$};
\node[draw,fill=lightgray,circle] (n14)  at (4,11) {$p_{1,1}$};
\node[draw,fill=lightgray,circle] (n15)  at (4,10) {$p_{1,2}$};
\node[draw,fill=white,circle] (n101)  at (4,9) {$p_{1,3}$};
\node[draw,fill=lightgray,circle] (n16)  at (6,11) {$u_{2,1}$};
\node[draw,fill=lightgray,circle] (n17)  at (6,10) {$u_{2,2}$};
\node[draw,fill=white,circle] (n102)  at (6,9) {$u_{2,3}$};
\node[draw,fill=lightgray,circle] (n18)  at (10,11) {$p_{2,1}$};
\node[draw,fill=lightgray,circle] (n19)  at (10,10) {$p_{2,2}$};
\node[draw,fill=white,circle] (n103)  at (10,9) {$p_{2,3}$};
\node[draw,fill=lightgray,circle] (n20)  at (12,11) {$u_{3,1}$};
\node[draw,fill=lightgray,circle] (n21)  at (12,10) {$u_{3,2}$};
\node[draw,fill=white,circle] (n104)  at (12,9) {$u_{3,3}$};
\node[draw,fill=lightgray,circle] (n22)  at (16,11) {$p_{3,1}$};
\node[draw,fill=lightgray,circle] (n23)  at (16,10) {$p_{3,2}$};
\node[draw,fill=white,circle] (n105)  at (16,9) {$p_{3,3}$};
\node[draw,fill=lightgray,circle] (n24)  at (1,12) {$z_{1}$};
\node[draw,fill=white,circle] (n25)  at (3,12) {$z'_{1}$};
\node[draw,fill=lightgray,circle] (n26)  at (7,12) {$z_{2}$};
\node[draw,fill=white,circle] (n27)  at (9,12) {$z'_{2}$};
\node[draw, fill=lightgray,circle] (n28)  at (13,12) {$z_{3}$};
\node[draw,fill=white,circle] (n29)  at (15,12) {$z'_{3}$};
\node[draw,fill=white,circle] (n30)  at (1,14) {$w_{1,1}$}; 
\node[draw,fill=white,circle] (n31)  at (7,14) {$w_{2,1}$}; 
\node[draw,fill=white,circle] (n32)  at (13,14) {$w_{3,1}$}; 
\node[draw,fill=lightgray,circle] (n33)  at (3,14) {$w'_{1,1}$}; 
\node[draw,fill=lightgray,circle] (n34)  at (9,14) {$w'_{2,1}$}; 
\node[draw,fill=lightgray,circle] (n35)  at (15,14) {$w'_{3,1}$}; 
\foreach \from/\to in {n100/n6, n101/n4, n102/n5, n103/n1, n104/n2, n105/n3,
n8/n9, n8/n10, n11/n7,n7/n6,n7/n3,n7/n1,n8/n6,n8/n3,n8/n1,n24/n6,n24/n4,n25/n6,
n25/n4, n26/n5, n26/n1, n27/n5, n27/n1, n28/n2,n28/n3, n29/n2, n29/n3, n12/n6,
n13/n6, n14/n4, n15/n4, n16/n5, n17/n5, n18/n1, n19/n1, n20/n2, n21/n2, n22/n3,
n23/n3, n30/n24, n31/n26, n32/n28, n33/n25, n34/n27, n35/n29}
\draw (\from) -- (\to);
\end{tikzpicture}
\caption{Predecessor configuration $Y'$ for the $G$ and $Y$ of Figure
\ref{fig:config_y} when all of $x_1, \neg{x_2}$ and $x_3$ are true. Shaded
circles indicate state $+1$ in $Y'$; empty circles indicate state $-1$.}
\label{fig:config_ylinha}
\end{figure}
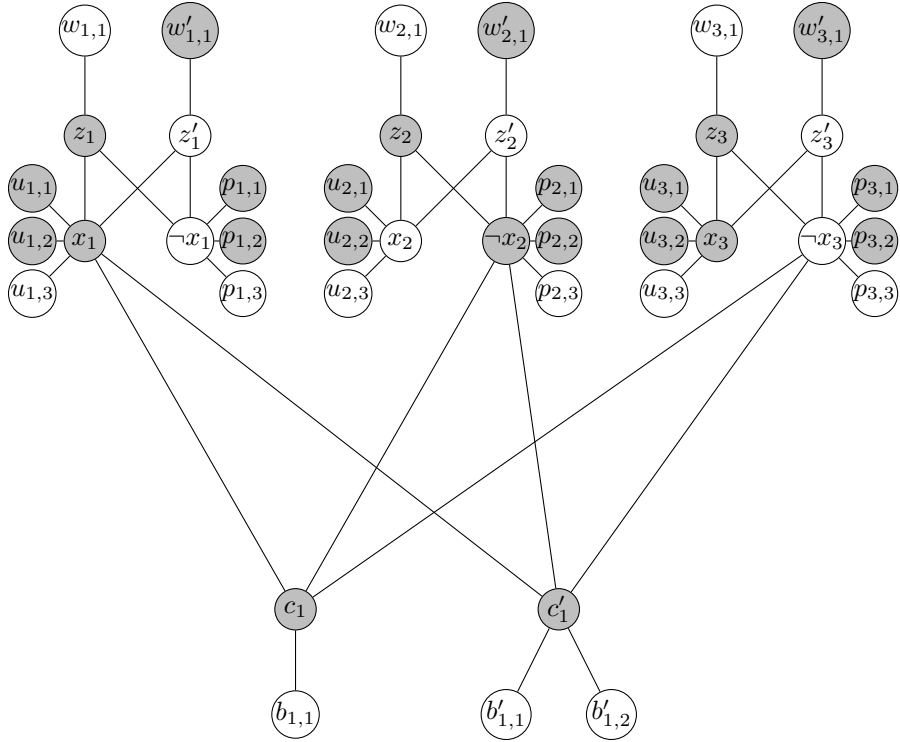

Hence, given the configuration $Y'$ as constructed, we see that the next
configuration is precisely configuration $Y$. We then conclude that whenever $S$
is satisfiable, $Y$ has at least one predecessor configuration.

Conversely, we now show that if $Y$ has at least one predecessor configuration
then $S$ is satisfiable.

In any predecessor configuration of $Y$, vertices $u_{i,j}$, $p_{i,j}$,
$w_{i,j}$, $w'_{i,j}$, $b_{i,j}$, and $b'_{i,j}$ should all be in the same
states as in configuration $Y$, since they all have only one neighbor each.

Vertices $z_{i}$ and $z'_{i}$ should also have the same states as in
configuration $Y$ in any predecessor configuration. For suppose that, in a
predecessor configuration, $z_{i}$ is in state $-1$; then necessarily $x_{i}$
and $\neg{x_{i}}$ should be in state $+1$, and consequently  vertex $z'_{i}$
would reach state $+1$ in the next configuration, which is different from its
state in configuration $Y$. An analogous argument holds for vertex $z'_{i}$.
This condition also forces $x_{i}$ to have a state different than $\neg{x_{i}}$
in any predecessor configuration of $Y$, since if both have the same state then
in the next configuration $z_{i}$ and $z'_{i}$ will also have the same state.

Each vertex $c_{i}$ must have state $+1$ in a predecessor configuration of $Y$,
otherwise every literal neighbor of $c_{i}$ would need to have state $-1$ in the
predecessor configuration and consequently vertex $c_i$ would not change to
state $+1$, which is its state in $Y$. Also, it must be the case that at least
two of the literal neighbors of $c_i$ have state $+1$, otherwise $c_{i}$ would
have state $-1$ in the next configuration, thus not matching configuration $Y$.

Each vertex $c'_{i}$ has the same literal neighbors as vertex $c_{i}$. As we
know that some of these neighbors have state $+1$ in a predecessor
configuration, $c'_{i}$ must have state $+1$, otherwise all of these neighbors
would need to have state $-1$ in the predecessor configuration, contradicting
the fact that some of them have state $+1$. Along with the restriction imposed
by $c_{i}$, we have that exactly two of the three literal vertices associated
with clause $c_{i}$ must have state $+1$ in a predecessor configuration.

Hence, given any predecessor configuration of $Y$, we can associate with any
literal the value true if its vertex has state $+1$, and value false otherwise.
The construction guarantees that this will satisfy every clause with exactly two
positive literals and that no opposite literals will have the same assignment. 

We conclude that if $Y$ has a predecessor configuration then $S$ is satisfiable.
\end{proof}

\begin{corolario}
\textsc{\textsc{Pre$(k)$}} on bipartite graphs is NP-complete for $k > 1$.
\label{cor:bipartido}
\end{corolario}

\begin{proof}
The graph constructed in the proof of Theorem \ref{np-completude} is bipartite
for the node sets
$\bigcup_{i,j} \{x_{i}, \neg{x_{i}}, b_{i,j}, b'_{i,j}, w_{i,j}, w'_{i,j}\}$
and $\bigcup_{i,j} \{c_{i}, c'_{i}, u_{i,j}, p_{i,j}, z_{i}, z'_{i} \}$.
\end{proof}

\section{Polynomial-time algorithms for trees}\label{trees}

In this section, we consider a tree $T$ rooted at an arbitrary vertex called
$\mathit{root}$. For each vertex $v$ in $T$, $\mathit{parent}_v$ denotes the
parent of $v$, and $\mathit{children}_v$ denotes the set of children of $v$. A
subtree of $T$ will be denoted by $T_{u}$, where $u$ is the root of the subtree.

It will be helpful to adopt a notation for a configuration of a subtree of $T$.
Denote by $Y_{v,t} = (Y_{t}(w_1), Y_{t}(w_2), \ldots, Y_{t}(w_{|T_v|}))$ the
configuration of states of the vertices in subtree $T_v$ at time $t$, where each
$w_i$ is one of the $|T_v|$ vertices of $T_v$. Notice that $Y_{v,t}$ is a
subsequence of $Y_{t}$ and its purpose is to refer to states of vertices in
subtree $T_v$ only; thus, if $v = \mathit{root}$, $Y_{v,t} = Y_t$.

The one-step dynamics in a $k$-reversible process in subtree $T_v$ can be
described using the function $F_{G,v}^{k} : Q^{|T_v|} \to Q^{|T_v|}$ such that
\begin{equation}
Y_{v,t} = F_{G,v}^k(Y_{v,t-1}) = (Y_{t}(w_1), Y_{t}(w_2), \ldots,
Y_{t}(w_{|T_v|})).
\end{equation}
As in Section $1$, we omit, for simplicity, the subscript referring to time,
using notation $Y'$ to refer to the configuration at time $t-1$ and $Y$ to
refer to the configuration at time $t$. Hence, $Y_v = Y_{v,t}$ and $Y'_v =
Y_{v,t-1}$.

So long as we take into account the influence of $\textit{parent}_v$ on the
dynamics of $T_v$, then it is easy to see that the following holds. If a
configuration $Y'$ exists for which $Y = F_{G}^{k}(Y')$, then we
have $Y_v = F_{G,v}^{k}(Y'_{v})$ as well. That is, the subsequence of $Y'$ that
corresponds to $T_v$ is a predecessor configuration of the subsequence of $Y$
that corresponds to $T_v$. 

Section \ref{sec:PRET} presents an algorithm that solves \textsc{Pre$(k)$} in
polynomial time for any $k$ when the graph is a tree. Section \ref{sec:PREC}
presents an algorithm that solves the associated counting problem
\textsc{\#Pre$(k)$}, also in polynomial time, for any $k$ when the graph is a
tree. 

\subsection{\boldmath Polynomial-time algorithm for \textsc{Pre$(k)$}}
\label{sec:PRET}

We start by defining the function $\mathit{vstate}(\mathit{target},
\mathit{current},p,k)$ that determines whether a vertex in state
$\mathit{current}$ can reach state $\mathit{target}$ in one time step assuming
that it has $p$ neighbors in the state different than $\mathit{current}$ in a
$k$-reversible process. This is a simple Boolean function that only checks
whether a given state transition is possible. It can be calculated in $O(1)$
time, since
\begin{equation}
\mathit{vstate}(\mathit{target},\mathit{current},p,k) =
\left\{
\begin{array}{rl}
\mathit{false},&\mbox{if } \mathit{current} =
\mathit{target} \mbox{ and } p \geq k, \\
&\mbox{or } \mathit{current} \neq \mathit{target} \mbox{ and } p < k;  \\
\mathit{true},&\mbox{otherwise.}
\end{array}
\right.
\end{equation}

The algorithm tries to build a predecessor configuration $Y'$ of $Y$ by
determining possible configurations for its subtrees and testing states on
vertices. Suppose we traverse the tree in a top-down fashion attaching states
to each vertex we visit. Let $v$ be the vertex of $T$ that the algorithm is
visiting. Suppose that for $parent_{v}$ the algorithm attached state $c$. We
define the function $\mathit{fstate}(v,c)$ that returns the state that $v$ must
have in a configuration $Y'_{v}$ such that $F(Y'_{v}) = Y_v$ or returns $\infty$
if there is no such configuration $Y'_{v}$ when $Y'(\mathit{parent}_v) = c$. If
both states are possible for $v$, the function simply returns
$Y(\mathit{parent}_v)$ or $+1$ when $v$ is the root. 

We assume that function $\mathit{fstate}$ will be called with parameters $v$ and
$0$  when $v$ is the root of the tree. Hence, $\mathit{fstate}(\mathit{root},0)$
is simply the state that $\mathit{root}$ must have in a predecessor
configuration of $Y$. If $\mathit{fstate}(\mathit{root},0)$ is different than
$\infty$, configuration $Y$ has a predecessor configuration; otherwise, it does
not.

Given the function $\mathit{fstate}(v,c)$ one can easily test whether the state
$c$ attached during the algorithm is possible or not in a predecessor
configuration of $Y$. When visiting vertex $\mathit{parent}_v$, the algorithm
calls the function $\mathit{fstate}(w,c)$ for each child $w$ of
$\mathit{parent}_v$, and then by using function $\mathit{vstate}$ it decides
whether a transition is possible from state $c$ to state $Y(\mathit{parent}_v)$.

Notice that when it is possible for both states to be returned by
$\mathit{fstate}(v,c)$, we define the function to return state
$Y(\mathit{parent}_v)$. This choice is always correct in this case, since it is
not important for the state to be actually assigned to $\mathit{parent}_v$ in
the predecessor configuration. What it does is to increase the number of
neighbors that can help $\mathit{parent}_v$ to reach state
$Y(\mathit{parent}_v)$ in the next time step. We maximize the number of
neighbors in state $Y(\mathit{parent}_v)$, and if this maximum number is not
enough to make $\mathit{parent}_v$ reach state $Y(\mathit{parent}_v)$, then a
lower number of neighbors would certainly not be either.

Now, the problem is to calculate $\mathit{fstate}(\mathit{root},0)$ in a correct
and efficient way. Assume that for each child $f$ of $v$ the values of
$\mathit{fstate}(f,+1)$ and $\mathit{fstate}(f,-1)$ are correctly calculated. In
other words, we know the states of every child of $v$ when $v$ has state $+1$ or
state $-1$ in a predecessor configuration. Let $l$ be the number of children of
$v$ with state other than the $Y'(v)$ calculated by the $\mathit{fstate}$
function. As we set the state of $\mathit{parent}_v$ we can update $l$ if
$\mathit{parent}_v$ also has a state different than $Y'(v)$, meaning that $l$
represents the number of vertices in a state different than $v$ in the
configuration we are trying to construct. Then we can verify
$\mathit{vstate}(Y(v),st,l,k)$ to check whether $\mathit{st}$ is a valid state
for $v$ in a predecessor configuration $Y'$ such that $F_{G,v}^k(Y'_v) = Y_v$.

An important observation to have in mind is that, for the case in which there is
at least one child $f$ of $v$ that does not have a possible state (which means
that $\mathit{fstate}(f,\mathit{st})$ = $\infty$), then the state $\mathit{st}$
is not valid for $v$.

Given that we already know all the valid states for $v$ when
$Y'(\mathit{parent}_v) = c$, function $\mathit{fstate}$ can be calculated like
this:
\begin{equation}
\mathit{fstate}(v,c) =
\left\{
\begin{array}{rl}
Y(\mathit{parent}_v),  &\mbox{if both states are valid for $v$;}  \\
+1,  &\mbox{if only state $+1$ is valid for $v$;}  \\
-1,  &\mbox{if only state $-1$ is valid for $v$;}  \\
\infty, & \mbox{otherwise.}
\end{array}
\right.
\end{equation}

Since we already know which the valid states are, $\mathit{fstate}(v,c)$ is
easily determined in $O(1)$ time. To check whether a given state is valid or not
we simply need to count the number of neighbors with state opposite to the one
being checked using function $\mathit{fstate}$, and then call function
$\mathit{vstate}$ to verify whether the state is valid. This can be done in
$O(d(v))$ time.

It is possible to calculate $\mathit{fstate}$ for all vertices in $T$ using a
recursive algorithm similar to depth-first search. The algorithm, when visiting
vertex $v$ with $\mathit{parent}_v$ in state $c$, recursively calculates
$\mathit{fstate}(f,+1)$ and $\mathit{fstate}(f,-1)$ for each child $f$ of $v$,
and once the algorithm returns from the recursion, all values needed to
calculate $\mathit{fstate}(v,c)$ are available. Hence, we simply need to
calculate $\mathit{fstate}(\mathit{root},0)$ recursively and check whether the
returned value is $\infty$ or not.

\begin{algorithm}[p]
    \label{alg:q1}
    \caption{$\mathit{calcfstate}(v,c)$}
    \BlankLine
    \Begin{
	\If{$\mathit{fstate}[v,c+1] \neq \mathit{NIL}$}
	{
	  \Return{$\mathit{fstate}[v,c+1]$}\;
	}
	$\mathit{count} \leftarrow 0$\;
	$\mathit{state} \leftarrow Y(\mathit{parent}_v)$\;
	\While{$\mathit{count} \neq 2$}
	{
	  $\mathit{count} \leftarrow \mathit{count} + 1$\;
	  $l \leftarrow  0$\;
	  $\mathit{ret} \leftarrow \mathit{true}$\;
	  \If{$c = -\mathit{state}$}
	  {
	    $l \leftarrow l + 1$\;
	  }  

	  \ForEach{$f \in \mathit{children}_v$}
	  {
	    \If{$\mathit{calcfstate}(f,\mathit{state}) = -\mathit{state}$}
	    {
		$l \leftarrow l + 1$\;
	    }
	    \If{$\mathit{calcfstate}(f,\mathit{state}) = \infty$}
	    {
		$\mathit{ret} \leftarrow \mathit{false}$\;
	    }
	  }
	  
	  \If{$\mathit{vstate}(Y(v), \mathit{state}, l, k) = \mathit{true} \mbox{ and } \mathit{ret} \neq \mathit{false}$}
	  {
	    $\mathit{fstate}[v,c+1] \leftarrow \mathit{state}$\;
	    \Return{$\mathit{state}$}\;
	  }
	  $\mathit{state} \leftarrow -\mathit{state}$\;
      }
      $\mathit{fstate}[v,c+1] \leftarrow \infty$\;
      \Return{$\infty$}\;
    }
\end{algorithm}

This is what Algorithm \ref{alg:q1} does. The algorithm maintains a table
$\mathit{fstate}$ that contains all the function values. This table is
initialized with the null value for all vertices and states.

For each vertex $v$, the algorithm tries to assign state $Y(\mathit{parent}_v)$,
and in case this is a valid state, this value is stored in the table and
returned. Otherwise, the algorithm tries to assign the opposite state to
$Y(\mathit{parent}_v)$ and does assign it in case it is a valid state. If none
of the states is valid then the algorithm stores $\infty$ in the table and
returns. Notice that when the algorithm accesses a position in the table it
increments the $c$ value by $1$, so it does not try to access a negative
position when $c$ has value $-1$. Table handling is very important to make the
algorithm run in polynomial time. Without verification in line 2, it would be a
simple backtracking algorithm with time complexity $O(2^n)$. We can verify
this by analyzing the case in which $T$ is the graph $P_n$ containing a single
path with $n$ vertices and $k=2$. Let $H(n)$ be the worst-case time complexity
of the algorithm in this case, without using the table. Suppose we choose the
root to be one of the vertices with degree $1$. $H(n)$ is easily verified to be
\begin{equation}
H(n) = 2H(n-1) + O(1),
\end{equation}
considering that the algorithm tries both states for each vertex, which happens
when the given input configuration does not have a predecessor configuration.
Additionally, the structure of $P_n$ allows us to calculate the number of steps
as a function of the number of steps to solve the problem on the subtree
$P_{n-1}$. By solving the above recurrence relation from $H(1) = O(1)$ we obtain
$H(n) = O(2^n)$. Using the table amounts to employing memoization \cite{cormen}
to avoid the exponential running time.

Suppose that the last function call is $\mathit{calcfstate}(v,c)$ and that this
is the first time the function is called with these parameters. In the worst
case, the algorithm calls $\mathit{calcfstate}(f,+1)$ and
$\mathit{calcfstate}(f,-1)$ for each child $f$ of $v$ and calculates
$\mathit{fstate}[v,c+1]$. For each child $f$ of $v$, the algorithm makes two
visits, and any other call to $\mathit{calcfstate}(v,c)$ will not result in any
other visit to $v$'s children since the algorithm returns
$\mathit{fstate}[v,c+1]$ in line 3. When the algorithm calls
$\mathit{calcfstate}(v,-c)$, again in the worst case it will call
$\mathit{calcfstate}(f,+1)$ and $\mathit{calcfstate}(f,-1)$ for each child $f$
of $v$, incrementing to four the number of visits to each child of $v$.
After this call to $\mathit{calcfstate}(v,-c)$, any other visit to vertex $v$
will not produce any other visit to any child of $v$ since both
$\mathit{fstate}(v,c)$ and $\mathit{fstate}(v,-c)$ will be stored in the table
and the algorithm will return in line 3. We conclude that each vertex will be
visited at most four times. We also conclude that each edge will be traversed at
most four times, twice for each state being tested. Thus, the time complexity of
Algorithm \ref{alg:q1} is $O(n+m)$. Using $m = n-1$ yields a time complexity of
$O(n)$.

Algorithm \ref{alg:q2} implements the recursive idea to recover a predecessor
configuration $Y'$ once we know that one exists. It traverses the tree accessing
values in table $\mathit{fstate}$ according to the vertex being visited and to
the state assigned to its parent. Once a state is assigned to a vertex, the
algorithm looks at the table to check the only option of recursive call to make.
Algorithm 2 traverses the tree exactly once, assigning states; thus, its time
complexity is $O(n)$.

\begin{algorithm}[t]
    \label{alg:q2}
    \caption{$\mathit{buildstate}(v,c)$}
    \BlankLine
    \Begin{
        $y[v] = \mathit{fstate}[v,c+1]$\;

	\ForEach{$f \in \mathit{children}_v$}
	 {
	    $\mathit{buildstate}(f,y[v])$\;
	 }
    }
\end{algorithm}

\subsection{\boldmath Polynomial-time algorithm for \textsc{\#Pre$(k)$}}
\label{sec:PREC}

Suppose a class of trees constructed in the following manner:
\begin{itemize}
\item A vertex $v_1$ connected to $p > 1$ vetices $v_2, v_3, \ldots, v_{p+1}$;
\item Each vertex $v_i$, with $2 \leq i \leq p+1$, connected to two other
vertices, denoted by $d_{i-1}$ and $e_{i-1}$. 
\end{itemize}

Each tree in this class has $3p + 1$ vertices. Assume a $2$-reversible process
and a configuration $Y$ in which all states are $+1$. In this case, any
configuration in which vertex $v_1$ has state $-1$ and at least two vertices out
of $v_2, v_3, \ldots, v_{p+1}$ have state $+1$ is a predecessor configuration of
$Y$. Notice that all vertices $d_i$ and $e_i$ must have state $+1$ in a
predecessor configuration of $Y$, since they have degree $1$. A lower bound on
the number of predecessor configurations of $Y$ for this class of trees is given
by
\begin{equation}
\sum_{i=2}^{p} {\binom{p}{i}}  = 2^{p} - p - 1.
\end{equation}

Figure \ref{fig:energiaEstrela} illustrates the construction of a tree in this
class for $p = 3$ and the respective configuration $Y$. We also illustrate some
of the possible predecessor configurations for a $2$-reversible process. 

Given this lower bound on the number of predecessor configurations, it turns out
that modifying the previous algorithm in order to store all possible predecessor
configurations and then reconstruct them is an exponential-time task. However,
solving the associated counting problem can be done in time $O(n^2)$ for every
$k$.

\subfiglabelskip=0pt
\begin{figure}[t]
\centering
\subfigure[]{
\begin{tikzpicture}
[scale=.8,auto=left,inner sep=0pt,minimum size=0.65cm]
\node[draw,fill=lightgray,circle] (n1)  at (0,0) {$v_{1}$};
\node[draw,fill=lightgray,circle] (n2)  at (2,0) {$v_{2}$};
\node[draw,fill=lightgray,circle] (n3)  at (0,3) {$v_{3}$};
\node[draw,fill=lightgray,circle] (n4)  at (-2,0) {$v_{4}$};
\node[draw,fill=lightgray,circle] (d1)  at (3,-1) {$d_{1}$};
\node[draw,fill=lightgray,circle] (e1)  at (3,1) {$e_{1}$};
\node[draw,fill=lightgray,circle] (d2)  at (-1,4) {$d_{2}$};
\node[draw,fill=lightgray,circle] (e2)  at (1,4) {$e_{2}$};
\node[draw,fill=lightgray,circle] (d3)  at (-3,-1) {$d_{3}$};
\node[draw,fill=lightgray,circle] (e3)  at (-3,+1) {$e_{3}$}; 
\foreach \from/\to in {n1/n2, n1/n3, n1/n4, n2/d1, n2/e1, n3/d2, n3/e2, n4/d3,
n4/e3}
\draw (\from) -- (\to);
\end{tikzpicture}
}
\subfigure[]{
\begin{tikzpicture}
[scale=.8,auto=left,inner sep=0pt,minimum size=0.65cm]
\node[draw,fill=white,circle] (n1)  at (0,0) {$v_{1}$};
\node[draw,fill=lightgray,circle] (n2)  at (2,0) {$v_{2}$};
\node[draw,fill=lightgray,circle] (n3)  at (0,3) {$v_{3}$};
\node[draw,fill=white,circle] (n4)  at (-2,0) {$v_{4}$};
\node[draw,fill=lightgray,circle] (d1)  at (3,-1) {$d_{1}$};
\node[draw,fill=lightgray,circle] (e1)  at (3,1) {$e_{1}$};
\node[draw,fill=lightgray,circle] (d2)  at (-1,4) {$d_{2}$};
\node[draw,fill=lightgray,circle] (e2)  at (1,4) {$e_{2}$};
\node[draw,fill=lightgray,circle] (d3)  at (-3,-1) {$d_{3}$};
\node[draw,fill=lightgray,circle] (e3)  at (-3,+1) {$e_{3}$}; 
\foreach \from/\to in {n1/n2, n1/n3, n1/n4, n2/d1, n2/e1, n3/d2, n3/e2, n4/d3,
n4/e3}
\draw (\from) -- (\to);
\end{tikzpicture}
}
\subfigure[]{
\begin{tikzpicture}
[scale=.8,auto=left,inner sep=0pt,minimum size=0.65cm]
\node[draw,fill=white,circle] (n1)  at (0,0) {$v_{1}$};
\node[draw,fill=lightgray,circle] (n2)  at (2,0) {$v_{2}$};
\node[draw,fill=white,circle] (n3)  at (0,3) {$v_{3}$};
\node[draw,fill=lightgray,circle] (n4)  at (-2,0) {$v_{4}$};
\node[draw,fill=lightgray,circle] (d1)  at (3,-1) {$d_{1}$};
\node[draw,fill=lightgray,circle] (e1)  at (3,1) {$e_{1}$};
\node[draw,fill=lightgray,circle] (d2)  at (-1,4) {$d_{2}$};
\node[draw,fill=lightgray,circle] (e2)  at (1,4) {$e_{2}$};
\node[draw,fill=lightgray,circle] (d3)  at (-3,-1) {$d_{3}$};
\node[draw,fill=lightgray,circle] (e3)  at (-3,+1) {$e_{3}$}; 
\foreach \from/\to in {n1/n2, n1/n3, n1/n4, n2/d1, n2/e1, n3/d2, n3/e2, n4/d3,
n4/e3}
\draw (\from) -- (\to);
\end{tikzpicture}
}
\caption{Configuration $Y$ with a possibly exponential number of predecessor
configurations. Shaded circles indicate state $+1$; empty circles indicate state
$-1$. (a) Configuration $Y$; (b) Predecessor configuration $Y'$; (c) Predecessor
configuration $Y''$.} \label{fig:energiaEstrela}
\end{figure}
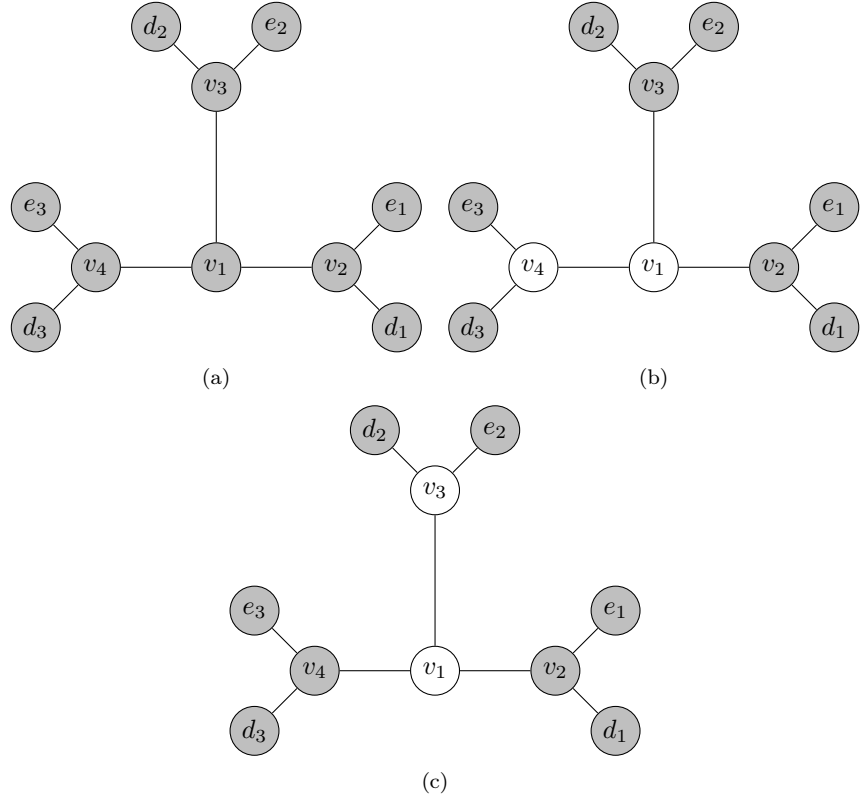

In order to do this, we use a function similar to $\mathit{fstate}$. However, a
more robust function must be used, one that will contain not only the state that
a vertex must have in a predecessor configuration, but also the number of
predecessor configurations of the subtree for the cases in which the vertex has
states $+1$ and $-1$. We define the function $\mathit{cfstate}(v,c)$ as an
ordered pair whose first element is the number of predecessor configurations of
subtree $T_v$ when $\mathit{parent}_v$ has state $c$ and $v$ has state $+1$, the
second element being the number of predecessor configurations of subtree $T_v$
when $\mathit{parent}_v$ has state $c$ and $v$ has state $-1$.

Similarly to function $\mathit{fstate}$, when $v$ is the root of the tree,
function $\mathit{cfstate}$ is called with parameters $v$ and $0$. Thus, the
total number of predecessor configurations of subtree $T_v$, when
$\mathit{parent}_v$ has state $c$, is the sum of the two elements in the ordered
pair $\mathit{cfstate}(v,c)$, and in case $v$ is the root, the sum of the two
elements in the ordered pair $\mathit{cfstate}(v, 0)$.

The natural way to calculate $\mathit{cfstate}(v,c)$ is quite simple. Without
loss of generality, suppose that we are calculating the first element of pair
$\mathit{cfstate}(v,c)$ and that, in this case, at least $l$ children of $v$
must have an arbitrary state $\mathit{st}$ in a predecessor configuration. For
simplicity, the first element in the pair $\mathit{cfstate}(v,c)$ will be
denoted by $\mathit{cfstate}(v,c)_{+1}$, whereas the second element will be
denoted by $\mathit{cfstate}(v,c)_{-1}$. Then
\begin{equation}
\mathit{cfstate}(v,c)_{+1}
= \sum\limits_{X \in C_{v}^{l}}^{}{\mathit{calc}(X,\mathit{st})},
\end{equation}
where $C_{v}^{l}$ is the set of all subsets of children of $v$ with at least $l$
elements and
\begin{equation}
\mathit{calc}(X,\mathit{st}) =
\left\{
\begin{array}{rl}
1,  &\mbox{if } X = \emptyset; \\
\prod\limits_{f \in X}^{}{\mathit{cfstate}(f,+1)_{\mathit{st}}}
\prod\limits_{f \in V(G) \setminus X}^{}{\mathit{cfstate}(f,+1)}_{-\mathit{st}},
&\mbox{otherwise. } \\
\end{array}
\right.
\end{equation}

In other words, we simply test all possibilities of state assignment to all
children of $v$ such that at least $l$ of them have state $+1$. For each one of
these possibilities we calculate the total number of predecessor configurations,
multiplying the number of predecessor configurations for each subtree. The value
$\mathit{cfstate}(v,c)_{+1}$ is the sum obtained in each possibility.

Notice that it is possible to calculate $\mathit{cfstate}(v,c)_{-1}$ likewise,
again assuming that in each predecessor configuration at least $l$ children of
$v$ have state $\mathit{st}$. We just need to redefine
$\mathit{calc}(X,\mathit{st})$ to be
\begin{equation}
\mathit{calc}(X,\mathit{st}) =
\left\{
\begin{array}{rl}
1,  &\mbox{if } X = \emptyset; \\
\prod\limits_{f \in X}^{}{\mathit{cfstate}(f,-1)_{\mathit{st}}}
\prod\limits_{f \in V(G) \setminus X}^{}{\mathit{cfstate}(f,-1)}_{-\mathit{st}},
&\mbox{otherwise. } \\
\end{array}
\right.
\end{equation}

A point to note in this approach is that the total number of configurations to
iterate through is $O(2^{d(v)-1})$. For example, for $l=1$,
\begin{equation}
|C_{v}^{1}| = \displaystyle\sum_{i=1}^{d(v) - 1}
{\binom{d(v) - 1}{i}} = 2^{d(v) - 1} - 1.
\end{equation}
We can, however, use dynamic programming \cite{cormen} to calculate
$\mathit{cfstate}(v,c)$ in polynomial time without needing to iterate through
all possible configurations.

Assume that the children of $v$ are ordered as in $\mathit{child}_{v,0}, \dots,
\mathit{child}_{v,d(v)-2}$, where $v$ is an internal vertex of the tree. If $v$
is the root, the order is: $\mathit{child}_{v,0}, \dots,\break
\mathit{child}_{v,d(v)-1}$. For simplicity, denote the number of children of
$v$ by $d'(v)$.

Define the function $g_{v}^{\mathit{rt}}(i,j)$ as the total number of
predecessor configurations for subtree $T_v$ in which $v$ has state
$\mathit{rt}$ and, moreover, exactly $j$ of the vertices $\mathit{child}_{v,0},
\mathit{child}_{v,1}, \ldots, \mathit{child}_{v,i-1}$ have
state $+1$. Similarly, define $h_{v}^{\mathit{rt}}(i,j)$ as the total number of
predecessor configurations for subtree $T_v$ in which $v$ has state
$\mathit{rt}$ and, moreover, exactly $j$ of the vertices $\mathit{child}_{v,0},
\mathit{child}_{v,1}, \ldots, \mathit{child}_{v,i-1}$ have
state $-1$. Then we can calculate $\mathit{cfstate}(v,c)_{\mathit{rt}}$ in the
following way.

If at least $l$ children of $v$ are required to have state $+1$ in a predecessor
configuration of $T_v$, then
\begin{equation}
\mathit{cfstate}(v,c)_{\mathit{rt}}
= \sum_{i=l}^{d'(v)} g_{v}^{\mathit{rt}}(d'(v),i),
\end{equation}
where $g_{v}^{\mathit{rt}}(i,j)$ is defined recursively, as in
\begin{equation}
g_{v}^{\mathit{rt}}(i,j) =
\left\{
\begin{array}{rl}
0,  &\mbox{if } i = 0 \mbox{ and } j > 0;\\
1,  &\mbox{if } i = 0 \mbox{ and } j = 0;\\
g_{v}^{\mathit{rt}}(i-1,j)a_i + g_{v}^{\mathit{rt}}(i-1,j-1)b_i,
&\mbox{if } i > 0 \mbox{ and } j > 0;\\
g_{v}^{\mathit{rt}}(i-1,j)a_i,  &\mbox{if } i > 0 \mbox{ and } j = 0,\\
\end{array}
\right.
\end{equation}
with $a_i = \mathit{cfstate}(\mathit{child}_{v,i}, \mathit{rt})_{-1}$ and $b_i
= \mathit{cfstate}(\mathit{child}_{v,i}, \mathit{rt})_{+1}$.

If, instead, at least $l$ children of $v$ are required to have state $-1$ in a
predecessor configuration of $T_v$, then
\begin{equation}
\mathit{cfstate}(v,c)_{\mathit{rt}}
= \sum_{i=l}^{d'(v)} h_{v}^{\mathit{rt}}(d'(v),i),
\end{equation}
where $h_{v}^{\mathit{rt}}(i,j)$ is such that
\begin{equation}
h_{v}^{\mathit{rt}}(i,j) =
\left\{
\begin{array}{rl}
0,  &\mbox{if } i = 0 \mbox{ and } j > 0;\\
1,  &\mbox{if } i = 0 \mbox{ and } j = 0;\\
h_{v}^{\mathit{rt}}(i-1,j)b_i + h_{v}^{\mathit{rt}}(i-1,j-1)a_i,
&\mbox{if } i > 0 \mbox{ and } j > 0;\\
h_{v}^{\mathit{rt}}(i-1,j)b_i,  &\mbox{if } i > 0 \mbox{ and } j = 0.\\
\end{array}
\right.
\end{equation}

As given above, the calculation of $g_{v}^{\mathit{rt}}(i,j)$ involves four
cases that depend on both state possibilities for vertex
$\mathit{child}_{v,i-1}$. They are the following (the cases for
$h_{v}^{\mathit{rt}}(i,j)$ are analogous):
\begin{itemize}
\item $i = 0$ and $j > 0$: In this case there is no vertex to be considered and
there should exist $j > 0$ vertices in state $+1$. Hence, no predecessor
configuration exists.
\item $i = 0$ and $j = 0$: This is the only case in which a predecessor
configuration exists with $i=0$, since there is no need to have a vertex in
state $+1$. Hence, there exists exactly one predecessor configuration.
\item $i > 0$ and $j > 0$: In this case we add the number of predecessor
configurations for the first $i$ subtrees, considering that vertex
$\mathit{child}_{v,i-1}$ has state $+1$, to the total number of predecessor
configurations when this vertex has state $-1$. Notice that if we assume that
$\mathit{child}_{v,i-1}$ has state $+1$, then exactly $j - 1$ of vertices
$\mathit{child}_{v,0}, \ldots, \mathit{child}_{v,i-2}$ must have state $+1$ as
well. Otherwise, if we assume $\mathit{child}_{v,i-1}$ to have state $-1$, then
$j$ of the first $i-1$ children of $v$ must have state $+1$.
\item $i > 0$ and $j = 0$: In this case none of vertices
$\mathit{child}_{v,0}, \ldots, \mathit{child}_{v,i-1}$ is required to have state
$+1$. We calculate the total number of predecessor configurations of the
subtrees rooted at vertices
$\mathit{child}_{v,0}, \ldots, \mathit{child}_{v,i-2}$ with all of them having
state $-1$, which is given by $g_{v}^{\mathit{rt}}(i-1,j)$, and multiply it
by the total number of predecessor configurations of subtree
$T_{\mathit{child}_{v,i-1}}$ with vertex $\mathit{child}_{v,i-1}$ having state
$-1$, which is given by $\mathit{cfstate}(\mathit{child}_{v,i-1},
\mathit{rt})_{-1}$.
\end{itemize}

\begin{algorithm}[p]
\small
    \label{alg:q3}
    \caption{$\mathit{countfstate}(v,c)$}
    \BlankLine
    \Begin{
	\If{$\mathit{cfstate[v,c+1]} \neq \mathit{NIL}$}
	{
	  \Return{$\mathit{cfstate}[v,c+1]$}\;
	}
	$d'(v) \leftarrow d(v) -2$\;
	\If{$v = \mathit{root}$}
	{
	  $d'(v) \leftarrow d(v) - 1$\;
	}
	$\mathit{count} \leftarrow 0$; $\mathit{state} \leftarrow +1$\;
	\While{$\mathit{count} \neq 2$}
	{
	  $i \leftarrow 0$; $\mathit{count} \leftarrow \mathit{count} + 1$\;
	  \ForEach{$f \in \mathit{parent}_v$}
	  {
	    $\mathit{par}[i] \leftarrow \mathit{countfstate}(f,\mathit{state})$; $i \leftarrow i + 1$\;
	  }
	  $l \leftarrow \mathit{threshold}_v(\mathit{state},c,k)$\;
	  \If{$v = \mathit{root}$}
	  {
	    $l \leftarrow \mathit{threshold}_v(\mathit{state},\infty,k)$\;
	  }

	  $\mathit{tab}[0,0] \leftarrow 1$\;
	
	  \For{$i \leftarrow 1$ to $d'(v)$}
	  {
	       $\mathit{tab}[0,j] \leftarrow 0$\;
	  }

	  \For{$i \leftarrow 1$ to $d'(v)$}
	  {
	      \For{$j \leftarrow 0$ to $d'(v)$}
	      {
		 \If{$Y(v) = +1$}
		 {
		    $\mathit{tab}[i,j] \leftarrow \mathit{tab}[i-1,j] \mathit{par}[i]_{-}$\;
		    \If{$j > 0$}
		    {
		      $\mathit{tab}[i,j] \leftarrow \mathit{tab}[i,j] + \mathit{tab}[i-1,j-1] \mathit{par}[i]_{+}$\;
		    }
		 }
		 \Else
		 {
		    $\mathit{tab}[i,j] \leftarrow \mathit{tab}[i-1,j] \mathit{par}[i]_{+}$\;
		    \If{$j > 0$}
		    {
		      $\mathit{tab}[i,j] \leftarrow \mathit{tab}[i,j] + \mathit{tab}[i-1,j-1] \mathit{par}[i]_{-}$\;
		    }
		 }
	      }
	  }
	  $\mathit{cfstate}[v,c+1]_\mathit{state} \leftarrow \sum_{i=l}^{d'(v)} \mathit{tab}[d'(v),i]$\;
	  $\mathit{state} \leftarrow -\mathit{state}$\;
      }
      \Return{$\mathit{cfstate}[v,c+1]$}\;
    }
\end{algorithm}

Applying the recursion directly results in an algorithm whose number of
operations grows very fast. Thus, once again we resort to dynamic programming to
calculate $g_{v}^{\mathit{rt}}$ and $h_{v}^{\mathit{rt}}$.

Algorithm \ref{alg:q3} implements the recursive scheme given above. Similarly
to Algorithm \ref{alg:q1}, we keep $\mathit{cfstate}$ values in a table to avoid
exponential times. The algorithm does not use the functions
$g_{v}^{\mathit{rt}}$ and $h_{v}^{\mathit{rt}}$ explicitly, but instead a table
$tab$ for checking the value of $Y(v)$ to decide which multiplication to perform
in lines 23 and 27. We use $g_{v}^{\mathit{rt}}$ if $Y(v)=+1$ and
$h_{v}^{\mathit{rt}}$ if $Y(v)=-1$. Suppose that $Y(v)=+1$.
If in a predecessor configuration of $Y$ vertex $v$ has state $-1$,
then, depending on the state of $\mathit{parent}_v$, vertex $v$ will need $k$ or
$k-1$ children with state $+1$ in that configuration. If $v$ has
state $+1$ in the predecessor configuration, then, depending on the state of
$\mathit{parent}_v$, vertex $v$ will need $d(v) - k + 1$ or $d(v) - k$ children
with state $+1$. The analysis for the case of $Y(v) = -1$ is analogous.

Besides choosing which function to use, we also need to calculate the value of
$l$. Given a vertex $v$ and that $\mathit{parent}_v$ has state $c$ in
the predecessor configuration, we define the function
$\mathit{threshold}_v(\mathit{current},c,k)$ as the least number of children of
$v$ in state $Y(v)$ in the predecessor configuration such that in the next step
$v$ has state $Y(v)$, assuming that $v$ has state \textit{current} in the
predecessor configuration. Thus,
\begin{equation}
\mathit{threshold}_v(\mathit{current},c,k) =
\left\{
\begin{array}{rl}
\min\{d(v) - k + 1,0\},  &\mbox{if } Y(v) = \mathit{current}\\
&\mbox{and } c \neq Y(v);\\
\min\{d(v) - k,0\},  &\mbox{if } Y(v) = \mathit{current}\\
&\mbox{and } c = Y(v);\\
k,  &\mbox{if } Y(v) \neq \mathit{current}\\
&\mbox{and } c \neq Y(v);\\
k - 1, &\mbox{if } Y(v) \neq \mathit{current}\\
&\mbox{and } c = Y(v).\\
\end{array}
\right.
\end{equation}

The time spent for each vertex $v$ in the double loop of line $18$ is
$O(d'(v)^2)$. Summing up over all vertices we get an $O(n^2)$ time complexity.

\section{\boldmath Polynomial-time algorithm for \textsc{Pre$(2)$}\hfill\break
on graphs with maximum degree no greater than 3}
\label{bounded}

In this section, we show that \textsc{Pre$(2)$} is in P when $\Delta(G) \leq 3$,
where $\Delta(G)$ is the maximum degree in $G$. Hence, this result covers the
case of cubic graphs.

We show how to reduce the problem to \textsc{2Sat}, solvable in $O(N+M)$ time;
as before, $N$ is the number of variables and $M$ is the number of
clauses \cite{2satP}. That is, given a configuration $Y$ we want to create a
\textsc{2Sat} instance $S$ such that $S$ is satisfiable if and only if $Y$ has a
predecessor configuration.

We start creating $S$ by adding literals $x_v$ and $\neg{x_v}$ for each vertex
$v$ in the graph. We will construct the clauses of $S$ in such a way that
whenever $Y$ has a predecessor configuration $Y'$, $S$ is satisfied by
letting each $x_v$ with $Y'(v) = +1$ be true and each $x_v$ with $Y'(v) = -1$
be false. Similarly, from any satisfying truth assignment for $S$ we construct a
predecessor configuration of $Y$ by assigning state $+1$ to $v$ whenever
$x_v$ is true and assigning state $-1$ whenever $x_v$ is false. Because
$\Delta(G) \leq 3$, we can construct a set of clauses in the following way.

For each vertex $v$ such that $Y(v) = +1$:
\begin{itemize}
\item If $d(v) = 1$: In the predecessor configuration $Y'$, $v$ must have the
same state as in configuration $Y$, since the process is $2$-reversible. Thus,
we add the clause:
\begin{itemize}
\item $x_v$.
\end{itemize}
Clearly, this clause is satisfied whenever $Y$ has a predecessor configuration.
\item If $d(v) = 2$ with neighbors $u$ and $w$: If in the predecessor
configuration $Y'$ vertex $v$ has state $+1$, then in $Y'$ at least one of its
neighbors must also have state $+1$. If in the predecessor configuration $Y'$
vertex $v$ has state $-1$, then both its neighbors must have state $+1$. We can
encode these conditions by adding the following clauses:
\begin{itemize}
\item $\neg{x_v} \rightarrow x_u \equiv x_v \vee x_u;$
\item $\neg{x_v} \rightarrow x_w \equiv x_v \vee x_w;$
\item $x_v \rightarrow x_u \vee x_w \equiv \neg{x_v} \vee x_u \vee x_w.$
\end{itemize}
Analyzing these three clauses reveals that when $x_v$ is true then $x_u$ or
$x_w$ is also true in order to make all three clauses satisfiable. In case
$x_v$ is false, we force $x_u$ and $x_w$ to have value true. We simplify the
clauses as follows:
\begin{itemize}
\item $x_v \vee x_u$;
\item $x_v \vee x_w$;
\item $x_u \vee x_w$.
\end{itemize} 
\item If $d(v)$ = 3 with neighbors $u$, $w$, and $z$: If in the predecessor
configuration $Y'$ vertex $v$ has state $+1$, then at least two of its neighbors
must have state $+1$ in $Y'$. If in $Y'$ vertex $v$ has state $-1$, then again
at least two of its neighbors must have state $+1$ in $Y'$. Therefore, we add
the following clauses:
\begin{itemize}
\item $\neg{x_v} \rightarrow x_u \vee x_w \equiv x_v \vee x_u \vee x_w$;
\item $\neg{x_v} \rightarrow x_u \vee x_z \equiv x_v \vee x_u \vee x_z$;
\item $\neg{x_v} \rightarrow x_w \vee x_z \equiv x_v \vee x_w \vee x_z$;
\item $x_v \rightarrow x_u \vee x_w \equiv \neg{x_v} \vee x_u \vee x_w$;
\item $x_v \rightarrow x_u \vee x_z \equiv \neg{x_v} \vee x_u \vee x_z$;
\item $x_v \rightarrow x_w \vee x_z \equiv \neg{x_v} \vee x_w \vee x_z$.
\end{itemize} 
As the value assigned to $x_v$ is not important in this subset of clauses, we
can easily simplify them:
\begin{itemize}
\item $x_u \vee x_w$;
\item $x_u \vee x_z$;
\item $x_w \vee x_z$.
\end{itemize}
\end{itemize}

For each vertex $v$ such that  $Y(v) = -1$, the cases are analogous:
\begin{itemize}
\item If $d(v) = 1$, create the clause:
\begin{itemize}
\item $\neg{x_v}$.
\end{itemize}
\item If $d(v) = 2$ with neighbors $u$ and $w$, create the clauses:
\begin{itemize}
\item $\neg{x_v} \vee \neg{x_u}$;
\item $\neg{x_v} \vee \neg{x_w}$;
\item $\neg{x_u} \vee \neg{x_w}$.
\end{itemize}
\item If $d(v) = 3$ with neighbors $u$, $w$, and $z$, create the clauses:
\begin{itemize}
\item $\neg{x_u} \vee \neg{x_w}$;
\item $\neg{x_u} \vee \neg{x_z}$;
\item $\neg{x_w} \vee \neg{x_z}$.
\end{itemize} 
\end{itemize}

We have constructed the set of clauses for $S$ in such a way that it is directly
satisfiable if $Y$ has at least one predecessor configuration. It now remains
for us to argue that the configuration $Y'$ obtained from a satisfying truth
assignment as explained above is indeed a predecessor of $Y$.

Suppose, to the contrary, that such a $Y'$ is not a predecessor of $Y$. In other
words, at least one vertex $v$ exists that does not reach state $Y(v)$ within
one time step. Analyzing the case of $Y(v) = +1$ we have the following
possibilities:
\begin{itemize}
\item $Y'(v) = +1$: In order to force $v$ to change its state $v$ must have at
least two neighbors in state $-1$. Hence, $v$ is necessarily a vertex with
degree at least $2$.

If $d(v) = 2$, then both neighbors of $v$ have state $-1$ and we have
clause $x_u \vee x_w$ not satisfied, which is a contradiction.

If $d(v) = 3$, then in a similar way there is an unsatisfied clause.

\item $Y'(v) = -1$: In order to force $v$ to remain in state $-1$ in the next
time step, there must be at most one neighbor in state $+1$. Notice that we must
have $d(v) \neq 1$, since $Y(v) = +1$, and then the satisfied clause
$x_v$ forces $Y'(v) = +1$.

If $d(v) = 2$ with neighbors $u$ and $w$, then at least one of these two
vertices must have state $-1$; hence $v$ does not change its state to $+1$, but
this implies that one of the clauses $x_v \vee x_w$ or $x_v \vee x_u$ it not
satisfied. 

If $d(v) = 3$, then since we have clauses with two literals involving all
the three neighbors of $v$, at least one of them is not satisfied.
\end{itemize}

The case of $Y(v) = -1$ is analogous. We conclude that if $S$ is satisfiable
then configuration $Y'$ is a predecessor configuration of $Y$.

To summarize, given a graph with $n$ vertices, we create $n$ variables and in
the worst case $3n$ clauses. Each clause has at most two literals. Thus, $S$ is
indeed a \textsc{2Sat} instance and we can solve \textsc{Pre(2)} in $O(n)$ time
when $\Delta(G) \leq 3$.

\section{Conclusions}\label{concl}

In this paper we have dealt with \textsc{Pre($k$)} and \textsc{\#Pre($k$)}, our
denominations for the \textsc{Predecessor Existence} problem and its counting
variation for $k$-reversible processes. We have shown that \textsc{Pre($1$)} is
solvable in polynomial time, that \textsc{Pre($k$)} is NP-complete for $k > 1$
even for bipartite graphs, and that it can be solved in
polynomial time for trees. For trees we have also shown that \textsc{\#Pre($k$)}
is polynomial-time solvable. We have also demonstrated the polynomial-time
solvability of \textsc{Pre($2$)} if $\Delta(G) \leq 3$.

We identify two problems worth investigating:
\begin{itemize}
\item Identify other cases in which \textsc{Pre($k$)} can be solved in
polynomial time.
\item Study the complexity properties of \textsc{\#Pre($2$)} for
$\Delta(G) \leq 3$.
\end{itemize}

\subsection*{Acknowledgments}

The authors acknowledge partial support from CNPq, CAPES, and FAPERJ BBP grants.

\bibliography{predex}
\bibliographystyle{plain}

\end{document}